\documentclass[12pt,letterpaper]{article}
\usepackage{amsmath,amsfonts,amsthm,amssymb}

\theoremstyle{plain}

\numberwithin{equation}{section}
\newtheorem{thm}{Theorem}[section]
\newtheorem{lem}[thm]{Lemma}
\newtheorem{cor}[thm]{Corollary}

\newcounter{cond}

\def\real{\mathbb R}
\newcommand{\ascript}{{\mathcal A}}
\newcommand{\bscript}{{\mathcal B}}
\newcommand{\fscript}{{\mathcal F}}
\newcommand{\kscript}{{\mathcal K}}
\newcommand{\mscript}{{\mathcal M}}
\newcommand{\oscript}{{\mathcal O}}
\newcommand{\tscript}{{\mathcal T}}
\newcommand{\vscript}{{\mathcal V}}

\newcommand{\mtilde}{\widetilde{M}}
\newcommand{\mhat}{\widehat{M}}
\newcommand{\nhat}{\widehat{N}}
\newcommand{\shat}{\widehat{s}}
\newcommand{\cupdot}{\dot{\cup}}
\newcommand{\rmtr}{\mathrm{tr}}

\newcommand{\ab}[1]{\left|#1\right|}
\newcommand{\doubleab}[1]{\left\|#1\right\|}
\newcommand{\brac}[1]{\left\{#1\right\}}
\newcommand{\paren}[1]{\left(#1\right)}
\newcommand{\sqbrac}[1]{\left[#1\right]}
\newcommand{\sqparen}[1]{{\left[#1\right)}}
\newcommand{\elbows}[1]{{\left\langle#1\right\rangle}}

\begin{document}

\title{COMPATIBILITY FOR\\
PROBABILISTIC THEORIES}
\author{Stan Gudder\\ Department of Mathematics\\
University of Denver\\ Denver, Colorado 80208\\
sgudder@du.edu}
\date{}
\maketitle

\begin{abstract}
We define an index of compatibility for a probabilistic theory (PT). Quantum mechanics with index 0 and classical probability theory with index 1 are at the two extremes. In this way, quantum mechanics is at least as incompatible as any PT. We consider a PT called a concrete quantum logic that may have compatibility index strictly between 0 and 1, but we have not been able to show this yet. Finally, we show that observables in a PT can be represented by positive, vector-valued measures.
\end{abstract}

\section{Observables in Probabilistic Theories}  
This paper is based on the stimulating article \cite{bhs12} by Busch, Heinosaari and Schultz. The authors should be congratulated for introducing a useful new tool for measuring the compatibility of a probabilistic theory (PT). In this paper, we present a simpler, but coarser, measure of compatibility that we believe will also be useful.

A \textit{probabilistic theory} is a $\sigma$-convex subset $\kscript$ of a real Banach space $\vscript$. That is, if 
$0\le\lambda _i\le 1$ with $\sum\lambda _i=1$ and $v_i\in\kscript$, $i=1,2,\ldots$, then $\sum\lambda _iv_i$ converges in norm to an element of $\kscript$. We call the elements of $\kscript$ \textit{states}. There is no loss of generality in assuming that $\kscript$ generates $\vscript$ in the sense that the closed linear hull of $\kscript$ equals $\vscript$. Denote the collection of Borel subsets of $\real ^n$ by $\bscript (\real ^n)$ and the set of probability measures on $\bscript (\real ^n)$ by
$\mscript (\real ^n)$. If $\kscript$ is a PT, an $n$-\textit{dimensional observable} on $\kscript$ is a $\sigma$-affine map
$M\colon\kscript\to\mscript (\real ^n)$. We denote the set of $n$-dimensional observables by $\oscript _n(\kscript )$ and write
$\oscript (\kscript )=\oscript _1(\kscript )$. We call the elements of $\oscript (\kscript )$ \textit{observables}. For
$M\in\oscript (\kscript )$, $s\in\kscript$, $A\in\bscript (\real )$, we interpret $M(s)(A)$ as the probability that $M$ has a value in $A$ when the system is in state $s$.

A set of observables $\brac{M_1,\ldots ,M_n}\subseteq\oscript (\kscript )$ is \textit{compatible} or \textit{jointly measurable} if there exists an $M\in\oscript _n(\kscript )$ such that for every $A\in\bscript (\real )$ and every $s\in\kscript$ we have
\begin{align*}
M&(s)(A\times\real\times\cdots\times\real )=M_1(s)(A)\\
M&(s)(\real\times A\times\real\times\cdots\times\real)=M_2(s)(A)\\
 \vdots&\\
M&(s)(\real\times\real\times\cdots\times\real\times A)=M_n(s)(A)
\end{align*}
In this case, we call $M$ a \textit{joint observable} for $\brac{M_1,\ldots ,M_n}$ and we call $\brac{M_1,\ldots ,M_n}$ the
\textit{marginals} for $M$. It is clear that if $\brac{M_1,\ldots ,M_n}$ is compatible, then any proper subset is compatible. However, we suspect that the converse is not true. If a set of observables is not compatible we say it is
\textit{incompatible}.

It is clear that convex combinations of observables give an observable so $\oscript (\kscript )$ forms a convex set. In the same way, $\oscript _n(\kscript )$ is a convex set. Another way of forming new observables is by taking functions of an observable. If $f\colon\real\to\real$ is a Borel function and $M\in\oscript (\kscript )$, the observable $f(M)\colon\kscript\to\mscript (\real )$ is defined by $f(M)(s)(A)=M(s)\paren{f^{-1}(A)}$ for all $s\in\kscript$, $A\in\bscript (\real )$.

\begin{thm}    
\label{thm11}
If $M_1,M_2\in\oscript (\kscript )$ are functions of a single observable $M$, then $M_1$, $M_2$ are compatible.
\end{thm}
\begin{proof}
Suppose $M_1=f(M)$, $M_2=g(M)$ where $f$ and $g$ are Borel functions. For $A,B\in\bscript (\real )$, $s\in\kscript$ define
$\mtilde (s)$ on $A\times B$ by
\begin{equation*}
\mtilde (s)(A\times B)=M(s)\sqbrac{f^{-1}(A)\cap g^{-1}(B)}
\end{equation*}
By the Hahn extension theorem, $\mtilde (s)$ extends to a measure in $\mscript (\real ^2)$. Hence,
$\mtilde\in\oscript _2(\kscript )$ and the marginals of $\mtilde$ are $f(M)$ and $g(M)$. We conclude that $M_1=f(M)$ and $M_2=g(M)$ are compatible
\end{proof}

It follows from Theorem~\ref{thm11} that an observable is compatible with any Borel function of itself and in particular with itself. In a similar way we obtain the next result.

\begin{thm}    
\label{thm12}
If $M_1,M_2\in\oscript (\kscript )$ are compatible and $f$, $g$ are Borel functions, then $f(M_1)$ and $g(M_2)$ are compatible.
\end{thm}
\begin{proof}
Since $M_1$, $M_2$ are compatible, they have a joint observable $M\in\oscript _2(\kscript )$. For $A,B\in\bscript (\real )$,
$s\in\kscript$ define $\mtilde (s)$ on $A\times B$ by
\begin{equation*}
\mtilde (s)(A\times B)=M(s)\sqbrac{f^{-1}(A)\times g^{-1}(B)}
\end{equation*}
As in the proof of Theorem~\ref{thm11}, $\mtilde (s)$ extends to a measure in $\mscript (\real ^2)$. Hence, $\mtilde\in\oscript (\kscript )$ and the marginals of $\mtilde$ are
\begin{align*}
\mtilde (s)(A\times\real )&=M(s)\sqbrac{f^{-1}(A)\times\real}=M_1(s)\sqbrac{f^{-1}(A)}=f(M_1)(s)(A)\\
\mtilde (s)(\real\times A)&=M(s)\sqbrac{\real\times g^{-1}(A)}=M_2(s)\sqbrac{g^{-1}(A)}=g(M_2)(s)(A)
\end{align*}
We conclude that $f(M_1)$ and $g(M_2)$ are compatible.
\end{proof}

The next result is quite useful and somewhat surprising.
\begin{thm}    
\label{thm13}
Let $M_i^j\in\oscript (\kscript )$ for $i=1,\ldots ,n$, $j=1,\ldots ,m$ and suppose $\brac{M_i^1,\ldots ,M_i^m}$ is compatible, $i=1,\ldots ,n$. If $\lambda _i\in\sqbrac{0,1}$ with $\sum\lambda _i=1$, $i=1,\ldots ,n$, then
\begin{equation*}
\brac{\sum _{i=1}^n\lambda _iM_i^1,\sum _{i=1}^n\lambda _iM_i^2,\ldots ,\sum _{i=1}^n\lambda _iM_i^m}
\end{equation*}
is compatible.
\end{thm}
\begin{proof}
Let $\mtilde _i\in\oscript _m(\kscript )$ be the joint observable for $\brac{M_i^1,\ldots ,M_i^m}$, $i=1,\ldots ,n$. Then
$\mtilde\!=\!\sum _{i=1}^n\lambda _i\mtilde _i$ is an $m$-dimensional observable with marginals

\begin{align*}
\mtilde&(s)(A\times\real\times\cdots\times\real)=\sum _{i=1}^n\lambda _i\mtilde _i(s)(A\times\real\times\cdots\times\real )
  =\sum _{i=1}^n\lambda _iM_i^1(s)(A)\\
  \mtilde&(s)(\real\times A\times\real\times\cdots\times\real )
  =\sum _{i=1}^n\lambda _i\mtilde _i(s)(\real\times A\times\real\times\cdots\times\real )\\
  &\hskip 11pc =\sum _{i=1}^n\lambda _iM_i^2(s)(A)\\
   \vdots&\\
  \mtilde&(s)(\real\times\real\times\cdots\times\real\times A)
  =\sum _{i=1}^n\lambda _i\mtilde _i(s)(\real\times\real\times\cdots\times\real\times A)\\
  &\hskip 11pc =\sum _{I=1}^n\lambda _iM_i^m(s)(A)
\end{align*}
The result now follows
\end{proof}

\begin{cor}    
\label{cor14}
Let $M,N,P\in\oscript (\kscript )$ and $\lambda\in\sqbrac{0,1}$. If $M$ is compatible with $N$ and $P$, then $M$ is compatible with $\lambda N+(1-\lambda )P$.
\end{cor}
\begin{proof}
Since $\brac{M,N}$ and $\brac{M,P}$ are compatible sets, by Theorem~\ref{thm13}, we have that
$M=\lambda M+(1-\lambda )M$ is compatible with $\lambda N+(1-\lambda )P$.
\end{proof}

\section{Noisy Observables} 
If $p\in\mscript (\real )$, we define the \textit{trivial observable} $T_p\in\oscript (\kscript )$ by $T_p(s)=p$ for every $s\in\kscript$. A trivial observable represents noise in the system. We denote the set of trivial observables on $\kscript$ by $\tscript (\kscript )$. The set $\tscript (\kscript )$ is convex with
\begin{equation*}
\lambda T_p+(1-\lambda )T_q=T_{\lambda p+(1-\lambda )q}
\end{equation*}
for every $\lambda\in\sqbrac{0,1}$ and $p,q\in\mscript (\real )$. An observable $M\in\oscript (\kscript )$ is compatible with any $T_p\in\tscript (\kscript )$ and a joint observable $\mtilde\in\oscript _2(\kscript )$ is given by
\begin{equation*}
\mtilde (s)(A\times B)=p(A)M(s)(B)
\end{equation*}
If $M\in\oscript (\kscript )$, $T\in\tscript (\kscript )$ and $\lambda\in\sqbrac{0,1}$ we consider $\lambda M+(1-\lambda )T$ as the observable $M$ together with noise. Stated differently, we consider $\lambda M+(1-\lambda )T$ to be a noisy version of $M$. The parameter $1-\lambda$ gives a measure of the proportion of noise and is called the \textit{noise index}. Smaller $\lambda$ gives a larger proportion of noise. As we shall see, incompatible observables may have compatible noisy versions.

The next lemma follows directly from Corollary~\ref{cor14}. It shows that if $M$ is compatible with $N$, then $M$ is compatible with any noisy version of $N$.

\begin{lem}    
\label{lem21}
If $M\in\oscript (\kscript )$ is compatible with $N\in\oscript (\kscript )$, then $M$ is compatible with $\lambda N+(1-\lambda )T$ for any $\lambda\in\sqbrac{0,1}$ and $T\in\tscript (\kscript )$.
\end{lem}

The following lemma shows that for any $M,N\in\oscript (\kscript )$ a noisy version of $N$ with noise index $\lambda$ is compatible with any noisy version of $M$ with noise index $1-\lambda$. The lemma also shows that if $M$ is compatible with a noisy version of $N$, then $M$ is compatible with a still noisier version of $N$.

\begin{lem}    
\label{lem22}
Let $M,N\in\oscript (\kscript )$ and $S,T\in\tscript (\kscript )$.
{\rm (a)}\enspace If $\lambda\in\sqbrac{0,1}$, then $\lambda M+(1-\lambda )T$ and $(1-\lambda )N+\lambda S$ are compatible.
{\rm (b)}\enspace If $M$ is compatible with $\lambda N+(1-\lambda )T$, then $M$ is compatible with $\mu N+(1-\mu )T$ where 
$0\le\mu\le\lambda\le 1$.
\end{lem}
\begin{proof}
(a)\enspace Since $\brac{M,S}$ and $\brac{T,N}$ are compatible sets, by Theorem~\ref{thm13}
$\lambda M+(1-\lambda )T$ is compatible with $\lambda S+(1-\lambda )N$.
(b)\enspace We can assume that $\lambda >0$ and we let $\alpha =\mu /\lambda$ so $0\le\alpha\le 1$. Since
$\brac{M,\lambda N+(1-\lambda )T}$ and $\brac{M,T}$ are compatible sets, by Theorem~\ref{thm13},
$M=\alpha M+(1-\alpha )M$ is compatible with
\begin{align*}
\alpha\sqbrac{\lambda N+(1-\lambda )T}+(1-\alpha )T&=\alpha\lambda N+\sqbrac{\alpha (1-\lambda )+(1-\alpha )}T\\
   &=\mu N+(1-\mu )T\qedhere
\end{align*}
\end{proof}

The \textit{compatibility region} $J(M_1,M_2,\ldots ,M_n)$ of observables $M_i\in\oscript (\kscript )$, $i=1,\ldots ,n$, is the set of points $(\lambda _1,\lambda _2,\ldots ,\lambda _n)\in\sqbrac{0,1}^n$ for which there exist $T_i\in\tscript (\kscript )$, $i=1,2,\ldots ,n$, such that
\begin{equation*}
\brac{\lambda _iM_i+(1-\lambda _i)T_i}_{i=1}^n
\end{equation*}
form a compatible set. Thus, $J(M_1,M_2,\ldots ,M_n)$ gives parameters for which there exist compatible noisy versions of $M_1,M_2,\ldots ,M_n$. It is clear that $0=(0,\ldots ,0)\in J(M_1,M_2,\ldots ,M_n)$ and we shall show that $J(M_1,M_2,\ldots ,M_n)$ contains many points. We do not know whether $J(M_1,M_2,\ldots ,M_n)$ is\newline
symmetric under permutations of the $M_i$. For example, is $J(M_1,M_2)=J(M_2,M_1)$?

\begin{thm}    
\label{thm23}
$J(M_1,M_2,\ldots M_n)$ is a convex subset of $\sqbrac{0,1}^n$.
\end{thm}
\begin{proof}
Suppose $(\lambda _1,\ldots ,\lambda _n),(\mu _1,\ldots ,\mu _n)\in J(M_1,\ldots ,M_n)$. We must show that
\begin{align*}
\lambda (\lambda _1,\ldots ,\lambda _n)+&(1-\lambda )(\mu _1,\ldots ,\mu _n)\\
&=(\lambda\lambda _1+(1-\lambda )\mu _1,\ldots ,\lambda\lambda _n+(1-\lambda )\mu _n)\in J(M_1,\ldots ,M_n)
\end{align*}
for all $\lambda\in\sqbrac{0,1}$. Now there exist $S_1,\ldots ,S_n,T_1,\ldots ,T_n\in\tscript (\kscript )$ such that
$\brac{\lambda _iM_i+(1-\lambda _i)S_i}_{i=1}^n$ and $\brac{\mu _iM_i+(1-\mu _i)T_i}_{i=1}^n$ are compatible. By
Theorem~\ref{thm13} the set of observables
\begin{align*}
&\brac{\lambda\sqbrac{\lambda _iM_i+(1-\lambda _i)S_i}+(1-\lambda )\sqbrac{\mu _iM_i+(1-\mu _i)T_i}}\\
&\quad =\brac{(\lambda\lambda _i+(1-\lambda )\mu _i)M_i+\lambda (1-\lambda _i)S_i+(1-\lambda )(1-\mu _i)T_i}
\end{align*}
is compatible. Since
\begin{align*}
\lambda (1-\lambda _i)+(1-\lambda )(1-\mu _i)&=1-\lambda\lambda _i-\mu _i+\lambda\mu _i\\
  &=1-\sqbrac{\lambda\lambda _i+(1-\lambda )\mu _i}
\end{align*}
letting $\alpha _i=\lambda\lambda _i+(1-\lambda )\mu _i$ we have that
\begin{equation*}
U_i=\frac{1}{1-\alpha _i}\sqbrac{\lambda (1-\lambda _i)S_i+(1-\lambda )(1-\mu _i)T_i}\in\tscript (\kscript )
\end{equation*}
Since $\brac{\alpha _iM_i+(1-\alpha _i)U_i}_{i=1}^n$ forms a compatible set, we conclude that\newline
$(\alpha _1,\ldots ,\alpha _n)\in J(M_1,\ldots ,M_n)$.
\end{proof}

Let $\Delta _n=\brac{(\lambda _1,\ldots ,\lambda _n)\in\sqbrac{0,1}^n\colon\sum\lambda _i\le 1}$. To show that
$\Delta _n$ forms a convex subset of $\sqbrac{0,1}^n\subseteq\real ^n$, let
$(\lambda _1,\ldots ,\lambda _n),(\mu _1,\ldots ,\mu _n)\in\Delta _n$ and $\lambda\in\sqbrac{0,1}$. Then
$\lambda (\lambda _1,\ldots ,\lambda _n+(1-\lambda )(\mu _1,\ldots ,\mu _n)\in\sqbrac{0,1}^n$ and
\begin{equation*}
\sum _{i=1}^n\sqbrac{\lambda\lambda _i+(1-\lambda )\mu _i}
  =\lambda\sum\lambda _i+(1-\lambda )\sum\mu _i\le\lambda +(1-\lambda )=1
\end{equation*}

\begin{thm}    
\label{thm24}
If $\brac{M_1,\ldots ,M_n}\subseteq\oscript (\kscript )$, then $\Delta _n\subseteq J(M_1,\ldots ,M_n)$.
\end{thm}
\begin{proof}
Let $\delta _0=(0,0,\ldots ,0)\in\real ^n$, $\delta _i=(0,\ldots ,0,1,0,\ldots ,0)\in\real ^n$, $i=1,\ldots ,n$ where 1 is in the $i$th coordinate. It is clear that
\begin{equation*}
\delta _i\in J(M_1,\ldots ,M_n)\cap\Delta _n,\quad i=0,1,\ldots ,n
\end{equation*}
If $\lambda =(\lambda _1,\ldots ,\lambda _n)\in\Delta _n$, letting $\mu =\sum\lambda _i$ we have that $0\le\mu\le 1$,
$\sum\lambda _i+(1-\mu )=1$ and
\begin{equation*}
\lambda =\sum _{i=1}^n\lambda _i\delta _i+(1-\mu )\delta _0
\end{equation*}
It follows that $\Delta _n$ is the convex hull of $\brac{\delta _0,\delta _1,\ldots ,\delta _n}$. Since
\begin{equation*}
\brac{\delta _0,\delta _1,\ldots ,\delta _n}\subseteq J(M_1,\ldots ,M_n)
\end{equation*}
and $J(M_1,\ldots ,M_n)$ is convex, it follows that $\Delta _n\in J(M_1,\ldots ,M_n)$.
\end{proof}

The $n$-\textit{dimensional compatibility region} for PT $\kscript$ is defined by
\begin{equation*}
J_n(\kscript )=\cap\brac{J(M_1,\ldots ,M_n)\colon M_i\in\oscript (\kscript ),i=1,\ldots ,n}
\end{equation*}
We have that $\Delta _n\subseteq J_n(\kscript )\subseteq\sqbrac{0,1}^n$ and $J_n(\kscript )$ is a convex set that gives a measure of the incompatibility of observables on $\kscript$. As $J_n(\kscript )$ gets smaller, $\kscript$ gets more incompatible and the maximal incompatibility is when $J_n(\kscript )=\Delta _n$. For the case of quantum states $\kscript$, the set $J_2(\kscript )$ has been considered in detail in \cite{bhs12}.

We now introduce a measure of compatibility that we believe is simpler and easier to investigate than $J_2(M,N)$ For 
$M,N\in\oscript (\kscript )$, the \textit{compatibility interval} $I(M,N)$ is the set of $\lambda\in\sqbrac{0,1}$ for which there exists a $T\in\tscript (\kscript )$ such that $M$ is compatible with $\lambda N+(1-\lambda )T$. Of course, $0\in T(M,N)$ and $M$ and $N$ are compatible if and only if $1\in I(M,N)$. We do not know whether $I(M,N)=I(N,M)$. It follows from Lemma~\ref{lem22}(b) that if $\lambda\in T(M,N)$ and $0\le\mu\le\lambda$, then $\mu\in I(M,N)$. Thus, $I(M,N)$ is an interval with left endpoint 0. The \textit{index of compatibility} of $M$ and $N$ is $\lambda (M,N)=\sup\brac{\lambda\colon\lambda\in I(M,N)}$. We do not know whether $\lambda (M,N)\in I(M,N)$ but in any case $I(M,N)=\sqbrac{0,\lambda (M,N)}$ or $I(M,N)=\sqparen{0,\lambda (M,N)}$. For a PT $\kscript$, we define the \textit{interval of compatibility} for $\kscript$ to be
\begin{equation*}
I(\kscript )=\cap\brac{I(M,N)\colon M,N\in\oscript (\kscript )}
\end{equation*}
The \textit{index of compatibility} of $\kscript$ is
\begin{equation*}
\lambda (\kscript )=\inf\brac{\lambda (M,N)\colon M,N\in\oscript (\kscript )}
\end{equation*}
and $I(\kscript )=\sqbrac{0,\lambda (\kscript )}$ or $I(\kscript )=\sqparen{0,\lambda (\kscript )}$. Again, $\lambda (\kscript )=0$ gives a measure of incompatibility of the observables in $\oscript (\kscript )$.
\bigskip

\noindent\textbf{Example 1.}\ (Classical Probability Theory)\enspace Let $(\Omega ,\ascript )$ be a measurable space and let $\vscript$ be the Banach space of real-valued measures on $\ascript$ with the total variation norm. If $\kscript$ is the
$\sigma$-convex set of probability measures on $\ascript$, then $\kscript$ generates $\vscript$. There are two types of observables on $\kscript$, the \textit{sharp} and \textit{fuzzy} observables. The sharp observables have the form $M_f$ where $f$ is a measurable function $f\colon\Omega\to\real$ and $M_f(s)(A)=s\sqbrac{f^{-1}(A)}$. If $M_f$, $M_g$ are sharp observables, form the unique 2-dimensional observable $\mtilde$ satisfying
\begin{equation*}
\mtilde (s)(A\times B)=s\sqbrac{f^{-1}(A)\cap g^{-1}(B)}
\end{equation*}
Then $\mtilde$ is a joint observable for $M_f$, $M_g$ so $M_f$ and $M_g$ are compatible. The unsharp observables are obtained as follows. Let $\fscript (\Omega )$ be the set of measurable functions $f\colon\Omega\to\sqbrac{0,1}$. Let
$\mhat\colon\bscript (\real )\to\fscript (\Omega )$ satisfy $\mhat (\real )=1$, $\mhat (\cupdot A_i)=\sum\mhat (A_i)$. An unsharp observable has the form
\begin{equation*}
M(s)(A)=\int\mhat (A)ds
\end{equation*}
Two unsharp observables $M,N$ are also compatible because we can form the joint observable $\mtilde$ given by
\begin{equation*}
\mhat (S)(A\times B)=\int\mhat (A)\nhat (B)ds
\end{equation*}
We conclude that $J(\kscript )=\sqbrac{0,1}\times\sqbrac{0,1}$ and $I(\kscript )=\sqbrac{0,1}$ so $\kscript$ has the maximal amount of compatibility.
\bigskip

\noindent\textbf{Example 2.}\ (Quantum Theory)\enspace Let $H$ be a separable complex Hilbert space and let $\kscript$ be the $\sigma$-convex set of all trace 1 positive operators on $H$. Then $\kscript$ generates the Banach space of self-adjoint trace-class operators with the trace norm. It is well known that $M\in\oscript (\kscript )$ if and only if there exists a positive operator-valued measure (POVM) $P$ such that $M(s)(A)=\rmtr\sqbrac{sP(A)}$ for every $s\in\kscript$, $A\in\bscript (\real )$. It is shown in \cite{bhs12} that if $\dim H=\infty$, then there exist $M_1,M_2\in\oscript (\kscript )$ such that
$J_2(M_1,M_2)=\Delta _2$ and hence $J(\kscript )=\Delta _2$. If $\dim H<\infty$, then $J(\kscript )$ is not known, although partial results have been obtained and it is known that $J(\kscript )\to\Delta _2$ as $\dim H\to\infty$
\bigskip

Now let $H$ be an arbitrary complex Hilbert space with $\dim H\ge 2$. Although the Pauli matrices $\sigma _x$, $\sigma _y$ are 2-dimensional, we can extend them from a 2-dimensional subspace $H_0$ of $H$ to all of $H$ by defining
$\sigma _x\psi =0$ for all $\psi\in H_0^\perp$. Define the POVMs $M_x$, $M_y$ on $H$ by
$M_x(\pm 1)=\tfrac{1}{2}(I\pm\sigma _x)$, $M_y(\pm 1)=\tfrac{1}{2}(I\pm\sigma _y)$. It is shown in \cite{bhs12} that
\begin{equation*}
J(M_x,M_y)=\brac{(\lambda ,\mu )\in\sqbrac{0,1}\times\sqbrac{0,1}\colon\lambda ^2+\mu ^2\le 1}
\end{equation*}
Thus, $J(M_x,M_y)$ is a quadrant of the unit disk. We conclude that $M_x$ is compatible with $\mu M_y+(1-\mu )T$ for $T\in\tscript (\kscript )$ if and only if $1+\mu ^2\le 1$. Therefore, $\mu =0$, so $I(M_x,M_y)=\brac{0}$ and $\lambda (M_x,M_y)=0$. Thus, $I(\kscript )=\brac{0}$ and $\lambda (\kscript )=0$. We conclude that quantum mechanics has the smallest index of compatibility possible for a PT. The index of compatibility for a classical system is 1, so we have the two extremes. It would be interesting to find $\lambda (\kscript )$ for other PTs.

\section{Concrete Quantum Logics} 
We now consider a PT that seems to be between the classical and quantum PTs of Examples~1 and~2. A collection of subsets $\ascript$ of a set $\Omega$ is a $\sigma$-\textit{class} if $\emptyset\in\ascript$, $A^c\in\ascript$ whenever $A\in\ascript$ and if $A_i$ are mutually disjoint, $i=1,2,\ldots$, then $\cup A_i\in\ascript$. If $\ascript$ is a $\sigma$-class on $\Omega$, we call
$(\Omega ,\ascript )$ a \textit{concrete quantum logic}. A $\sigma$-\textit{state} on $\ascript$ is a map
$s\colon\ascript\to\sqbrac{0,1}$ such that $s(\Omega )=1$ and if $A_i\in\ascript$ are mutually disjoint, then
$s(\cup A_i)=\sum s(A_i)$. If $\kscript$ is the set of $\sigma$-states on $(\Omega ,\ascript )$, we call $\kscript$ a 
\textit{concrete quantum logic} PT. Let $\ascript _\sigma$ be the $\sigma$-algebra generated by $\ascript$. A $\sigma$-state
$s$ is \textit{classical} if there exists a probability measure $\mu$ on $\ascript _\sigma$ such that $s=\mu\mid\ascript$. As in the classical case, an observable is \textit{sharp} if it has the form $M_f(s)(A)=s\sqbrac{f^{-1}(A)}$ for an $\ascript$-measurable function $f\colon\Omega\to\real$. If $f$ and $g$ are $\ascript$-measurable functions satisfying
$f^{-1}(A)\cap g^{-1}(B)\in\ascript$ for all $A,B\in\bscript (\real )$, then $M_f$ and $M_g$ are compatible because they have a joint observable $M$ satisfying $M(s)(A\times B)=s\sqbrac{f^{-1}(A)\cap g^{-1}(B)}$ for all $s\in\kscript$, $A,B\in\bscript (\real )$. We do not know whether $M_f$ and $M_g$ compatible implies that $f^{-1}(A)\cap g^{-1}(B)\in\ascript$ holds for every
$A,B\in\bscript (\real )$, although we suspect it does not.
\bigskip

\noindent\textbf{Example 3.}This is a simple example of a concrete quantum logic. Let $\Omega =\brac{1,2,3,4}$ and let
$\ascript$ be the collection of subsets of $\Omega$ with even cardinality. Then
\begin{equation*}
\ascript =\brac{\emptyset ,\Omega ,\brac{1,2},\brac{3,4},\brac{1,3},\brac{2,4},\brac{1,4},\brac{2,3}}
\end{equation*}
Let $\kscript$ be the sets of all states on $\ascript$. Letting $a=\brac{1,2}$, $a'=\brac{3,4}$, $b=\brac{1,3}$, $b'=\brac{3,4}$,
$c=\brac{1,4}$, $c'=\brac{2,3}$ we can represent an $s\in\kscript$ by
\begin{align*}
\shat&=\paren{s(a),s(a'),s(b),s(b'),s(c),s(c')}\\
  &=\paren{s(a),1-s(a),s(b),1-s(b),s(c),1-s(c)}
\end{align*}
Thus, every $s\in\kscript$ has the form
\begin{equation*}
s=\paren{\lambda _1,1-\lambda _1,\lambda _2,1-\lambda _2,\lambda _3,1-\lambda _3}
\end{equation*}
for $0\le\lambda _i\le 1$, $i=1,2,3$. The pure (extremal) classical states are the 0-1 states:
$\delta _1=(1,0,1,0,1,0)$, $\delta _3=(1,0,0,1,0,1)$, $\delta _3=(0,1,1,0,0,1)$, $\delta _4=(0,1,0,1,1,0)$. The pure nonclassical states are the 0-1 states: $\gamma _1=1-\delta _1$, $\gamma _2=1-\delta _2$, $\gamma _3=1-\delta _3$,
$\gamma _4=1-\delta _4$ where $1=(1,1,1,1,1,1)$. For example, to see that $\gamma _1$ is not classical, we have that
$\gamma _1=(0,1,0,1,0,1)$. Hence,
$\gamma _1\paren{\brac{3,4}}=\gamma _1\paren{\brac{2,4}}=\gamma _1\paren{\brac{2,3}}=1$. If there exists a probability measure $\mu$ such that $\gamma _1=\mu\mid\ascript$ we would have
$\mu\paren{\brac{1}}=\mu\paren{\brac{2}}=\mu\paren{\brac{3}}=\mu\paren{\brac{4}}=0$ which is a contradiction. The collection of sharp observable is very limited because a measurable function $f\colon\Omega\to\real$ can have at most two values. Thus, if $M_f$ is a sharp observable there exists $a,b\in\real$ such that $M_f(s)\paren{\brac{a,b}}=1$ for every $s\in\kscript$. There are many observables with more than two values (non-binary observables) and these are not sharp. Even for this simple example, it appears to be challenging to investigate the region and interval of compatibility.

\section{Vector-Valued Measures} 
Let $\kscript$ be a PT with generated Banach space $\vscript$ and $\vscript ^*$ be the Banach space dual of $\vscript$.
A \textit{normalized vector-valued measure} (NVM) for $\kscript$ is a map $\Gamma\colon\bscript (\real )\to\vscript ^*$ such that $A\mapsto\Gamma(A)(s)\in\mscript (\real )$ for every $s\in\kscript$. Thus, $\Gamma$ satisfies the conditions:
\begin{list} {(\arabic{cond})}{\usecounter{cond}%
\setlength\itemindent{-7pt}}
\item $\Gamma (\real )(s)=1$ for every $s\in\kscript$,
\item $0\le\Gamma (A)(s)\le 1$ for every $s\in\kscript$, $A\in\bscript (\real )$,
\item If $A_i\in\bscript (\real )$ are mutually disjoint, $i=1,2,\ldots$, then
\begin{equation*}
\Gamma (\cup A_i)(s)=\sum\Gamma (A_i)(s)
\end{equation*}
for every $s\in\kscript$.
\end{list}
This section shows that there is a close connection between observables on $\kscript$ and NVMs for $\kscript$.

\begin{thm}    
\label{thm41}
If $\Gamma$ is a NVM for $\kscript$, then $M\colon\kscript\to\mscript (\real )$ given by\newline
$M(s)(A)=\Gamma (A)(s)$, $s\in\kscript$, $A\in\bscript (\real )$, is an observable on $\kscript$.
\end{thm}
\begin{proof}
Since $A\mapsto\Gamma (A)(s)\in\mscript (\real )$ we have that $A\mapsto M(s)(A)\in\mscript (\real )$. Let
$\lambda _i\in\sqbrac{0,1}$ with $\sum\lambda _i=1$, $s_i\in\kscript$, $i=1,2,\ldots$, and suppose that $s=\sum\lambda _is_i$. Then $\lim\limits _{n\to\infty}\sum\limits _{i=1}^n\lambda _is_i=s$ in norm and since $s\mapsto\Gamma (A)(s)\in\vscript ^*$, for every $A\in\bscript (\real )$ we have
\begin{align*}
M(s)(A)&=M\paren{\sum\lambda _is_i}(A)=\Gamma (A)\paren{\sum\lambda _is_i}
  =\Gamma (A)\paren{\lim _{n\to\infty}\sum _{i=1}^n\lambda _is_i}\\
  &=\lim _{n\to\infty}\Gamma (A)\paren{\sum _{i=1}^n\lambda _is_i}=\lim _{n\to\infty}\sum _{i=1}^n\lambda _i\Gamma (A)(s_i)\\
  &=\lim _{n\to\infty}\sum _{i=1}^n\lambda _iM(s_i)(A)=\sum _{i=1}^\infty\lambda _iM(s_i)(A)
\end{align*}
It follows that $M\paren{\sum\lambda _is_i}=\sum\lambda _iM(s_i)$ so $M\in\oscript (\kscript )$.
\end{proof}

The converse of Theorem~\ref{thm41} holds if some mild conditions are satisfied. To avoid some topological and measure-theoretic technicalities, we consider the special case where $\vscript$ is finite-dimensional. Assuming that $\kscript$ is the base of a generating positive cone $\vscript ^+$, we have that every element $v\in\vscript ^+$ has a unique form $v=\alpha s$,
$\alpha \ge 0$, $s\in\kscript$ and that $\vscript =\vscript ^+\oplus\vscript ^-$ where $\vscript ^-=-\vscript ^+$ and
$\vscript ^+\cap\vscript ^-=\brac{0}$. If $M\in\oscript (\kscript )$, then for every $A\in\bscript (\real )$, $s\mapsto M(s)(A)$ is a convex, real-valued function on $\kscript$. A standard argument shows that this function has a unique linear extension
$\mhat (A)=\vscript ^*$ for every $A\in\bscript (\real )$. Hence
\begin{equation}         
\label{eq41}
\mhat (A)(s)=M(s)(A)
\end{equation}
for every $s\in\kscript$, $A\in\bscript (\real )$. Since $A\mapsto\mhat (A)(s)=M(s)(A)\in\mscript (\real )$ we conclude that
$A\mapsto\mhat (A)$ is a NVM and $\mhat$ is the unique NVM satisfying \eqref{eq41}. It follows that the converse of
Theorem~\ref{thm41} holds in this case.
\bigskip

\noindent\textbf{Example $1'$.}\ (Classical Probability Theory)\enspace In this example $\vscript ^*$ is the Banach space of bounded measurable functions $f\colon\Omega\to\real$ with norm $\doubleab{f}=\sup\ab{f(\omega )}<\infty$ and duality given by
\begin{equation*}
\elbows{\mu ,f}=f(\mu )=\int fd\mu
\end{equation*}
The function $1(\omega )=1$ for every $\omega\in\Omega$ is the natural unit satisfying $1(\mu )=1$ for every $\mu\in\kscript$. In this case, $\kscript$ is a base for the generating positive cone $\vscript ^+$ of bounded measures and the converse of Theorem~\ref{thm41} holds. Then a NVM $\Gamma$ has the form $0\le\Gamma (A)(\omega )\le 1$ for every
$A\in\bscript (\real )$, $\omega\in\Omega$ and $\Gamma (\real )=1$. Thus $\Gamma (A)\in\fscript (\Omega )$ and if $M$ is the corresponding observable, then
\begin{equation*}
M(\mu )(A)=\Gamma (A)(\mu )=\int\Gamma (A)d\mu
\end{equation*}
In particular, if $T_p\in\tscript (\kscript )$ then the corresponding NVM $\Gamma _p$ has the form
\begin{equation*}
\Gamma _p(A)(\mu )=T_p(\mu )(A)=p(A)
\end{equation*}
so $\Gamma _p(A)$ is the constant function $p(A)$. Moreover, if $M_p\in\oscript (\kscript )$ is sharp, then the corresponding NVM $\Gamma _f$ satisfies
\begin{equation*}
\int\Gamma _f(A)d\mu=\Gamma _f(A)(\mu )=M_f(\mu )(A)=\mu\sqbrac{f^{-1}(A)}=\int\chi _{f^{-1}(A)}d\mu
\end{equation*}
Hence, $\Gamma _f(A)=\chi _{f^{-1}(A)}$ for every $A\in\bscript (\real )$.
\bigskip

\noindent\textbf{Example $2'$.}\ (Quantum Theory)\enspace In this example $\vscript ^*$ is the Banach space $\bscript (H)$ of bounded linear operators on $H$ with norm
\begin{equation*}
\doubleab{L}=\sup\brac{\doubleab{L\psi}\colon\doubleab{\psi}=1}
\end{equation*}
and duality given by
\begin{equation*}
\elbows{s,L}=L(a)=\rmtr(sL)
\end{equation*}
The identity operator $I$ is the natural unit satisfying $I(s)=1$ for all $s\in\kscript$. In this case, $\kscript$ is a base for the generating cone $\vscript ^+$ of positive trace class operators and the converse of Theorem~\ref{thm41} holds, If $\Gamma$ is a NVM, then $\Gamma (A)$ is a positive operator satisfying $0\le\Gamma (A)\le I$ called an \textit{effect} and
$\Gamma (\real )=I$. According to the converse of Theorem~\ref{thm41}, if $M$ is an observable, then there exists a POVM
$\Gamma$ such that
\begin{equation*}
M(s)(A)=\rmtr\sqbrac{s\Gamma (A)}
\end{equation*}
for every $s\in\kscript$ and $A\in\bscript (\real )$. In particular, if $T_p\in\tscript (\kscript )$, then the corresponding NVM
$\Gamma _p$ has the form
\begin{equation*}
\rmtr\sqbrac{s\Gamma _p(A)}=\Gamma _p(A)(s)=T_p(s)(A)=p(A)=\rmtr\sqbrac{sp(A)I}
\end{equation*}
so $\Gamma _p(A)=p(A)I$ for all $A\in\bscript (\real )$.
\bigskip

Similar to a NVM, we define an $n$-\textit{dimensional} NVM to be a map\newline
$\Gamma\colon\bscript (\real ^n)\to\vscript ^*$ such that $A\mapsto\Gamma (A)(s)\in\mscript (\real ^b)$ for every $s\in\kscript$. Moreover, a set $\brac{\Gamma _1,\ldots ,\Gamma _n}$ of NVMs for $\kscript$ is compatible if there exists an $n$-dimensional NVM $\Gamma$ such that
\begin{align*}
\Gamma (A&\times\real\times\cdots\times\real )=\Gamma _1(A)\\
&\vdots\\
\Gamma (\real&\times\real\times\cdots\times\real\times A)=\Gamma _n(A)
\end{align*}
for every $A\in\bscript (\real )$. The proof of the following theorem is straightforward.
\begin{thm}    
\label{thm42}
If $\brac{M_1,\ldots ,M_n}\subseteq\oscript (\kscript )$ and $\brac{\Gamma _1,\ldots\Gamma _n}$ are the corresponding NVM for $\kscript$, then $\brac{M_1,\ldots ,M_n}$ are compatible if and only if $\brac{\Gamma _1,\ldots ,\Gamma _n}$ are compatible.
\end{thm}

\end{document}